\documentclass[runningheads,a4paper]{llncs}

\usepackage{amssymb}
\usepackage{comment}
\usepackage{enumerate}
\usepackage{bussproofs}
\usepackage{graphicx}
\usepackage{textcomp}
\usepackage{amssymb}
\usepackage{graphicx}
\usepackage{textcomp}
\usepackage{amsmath}
\usepackage{stmaryrd}
\usepackage{url}

\newcommand{\ri}{\RightLabel}
\newcommand{\f}{\footnotesize}
\newcommand{\p}[1]{[  #1 ]}

\newcommand{\msf}{\mathsf}

\newcommand{\mrm}{\mathrm}

\newcommand{\mc}{\mathcal}

\newcommand{\Imp}{\Rightarrow}

\newcommand{\GP}{G_\msf{PAL}}

\newtheorem{thm}{Theorem}
\newtheorem{defn}[thm]{Definition}
\newtheorem{lem}[thm]{Lemma}
\newtheorem{cor}[thm]{Corollary}
\newtheorem{prop}[thm]{Proposition}

\newtheorem{exam}[thm]{Example}
\newtheorem{rem}[thm]{Remark}
\usepackage{amssymb}
\usepackage{graphicx}

\newcommand{\keywords}[1]{\par\addvspace\baselineskip
\noindent\keywordname\enspace\ignorespaces#1}

\begin{document}
\mainmatter  

\title{A Labelled Sequent Calculus for Public Announcement Logic}

%
%
\author{Hao Wu\inst{1} \and
Hans van Ditmarsch\inst{2}\and
Jinsheng Chen\inst{3}}
\authorrunning{Wu et al.}
%
\institute{Institute of Logic and Cognition, Sun Yat-sen University \and
Department of Computer Science, Open University of the Netherlands
 \and
School of Philosophy, Zhejiang University}
%

%
%

\maketitle

\begin{abstract}
Public announcement logic ($\msf{PAL}$) is an extension of epistemic logic ($\msf{EL}$) with some reduction axioms. In this paper, we propose a cut-free labelled sequent calculus for $\msf{PAL}$, which is an extension of that for $\msf{EL}$ with sequent rules adapted from the reduction axioms. This calculus admits cut and allows terminating proof search. 

\end{abstract}

\section{Introduction}

The modern modal approach to the logic of knowledge and belief was extensively developed by Hintikka \cite{Hintikka1962} to interpret epistemic notions utilizing possible world semantics. The standard multi-agent epistemic logic $\msf{EL}$ is usually identified with the poly-modal modal logic $\msf{S5}$ with a group of agents. Public announcement logic (PAL), introduced in Plaza \cite{Plaza1989}, is an extension of $\msf{EL}$ that studies logical dynamics of epistemic information after the action of public announcement. More general actions are studied in action model logic (see e.g. \cite{vB2011,DIT2007,ALS1998}).

$\msf{PAL}$ is an extension of $\msf{EL}$ with announcement operators of the form $\p{\varphi}$. As $\msf{EL}$, formulas in $\msf{PAL}$ are interpreted in Kripke models in which all relations are reflexive, transitive and symmetric. In particular, formulas of the form $\p{\varphi}\psi$ are interpreted in the restrictions of Kripke models induced by the announcement $\varphi$. The Hilbert-style axiomatization of $\msf{PAL}$ is obtained by adding to that of $\msf{EL}$ the so called \emph{reduction axioms} for announcement operators, which can be used to eliminate announcement operators in a $\msf{PAL}$-formula.

Generally speaking, it is difficult to prove whether proof search using a Hilbert-style axiomatization is decidable. In view of these, many proof systems for $\msf{PAL}$ are proposed in the literature, e.g., tableau systems \cite{BAL2010}, labelled sequent calculi \cite{balbiani2014sequent,NOM2016}.

In this paper, we propose another labelled sequent calculus for $\msf{PAL}$. This calculus is based on a labelled sequent calculus for $\msf{EL}$ proposed in Hakli and Negri \cite{hakli2007proof} and rules for announcement operators designed according to the reduction axioms. This calculus admits structural rules (including cut) and allows terminating proof search. Unlike \cite{balbiani2014sequent,NOM2016}, which are based on another semantics for $\msf{PAL}$, our calculus is based on the original semantics and takes into account the conditions of reflexivity, transitivity and symmetry in $\msf{EL}$.

The rest of the paper is structured as follows: Section 2 presents some basic notions and concepts. Section 3 presents our labelled sequent calculus $G_\msf{PAL}$ for $\msf{PAL}$. Section 4 shows that $G_\msf{PAL}$ admits some structural rules, including cut. Section 5 shows that $G_\msf{PAL}$ allows terminating proof search.
Section 6 compares $G_\msf{PAL}$ with related works, concludes the paper and discusses future work.

\section{Preliminaries}

\subsection{$\msf{EL}$ and $\msf{PAL}$}\label{defn:preli}
Let $\msf{Prop}$ be a denumerable set of variables and $\msf{Ag}$ a finite set of agents. Language $\mc{L}_\msf{EL}$ for epistemic logic is defined inductively as follows:
\[
\mc{L}:: = p \mid \neg \varphi \mid \varphi \land \varphi \mid \varphi \to \varphi \mid K_a \varphi
\]
Moreover, language $\mc{L}_\msf{PAL}$ for public announcement logic is defined inductively as follows:
\[
\mc{L}:: = p \mid \neg \varphi \mid \varphi \land \varphi \mid \varphi \to \varphi \mid K_a \varphi \mid  \p{\varphi}\varphi
\]
where $p \in \msf{Prop}$ and $a \in \msf {Ag}$. We use $\varphi \leftrightarrow \psi$ as an abbreviation for $(\varphi \to \psi)\land (\psi \to \varphi)$. $\mc{L}_\msf{PAL}$ is an extension of $\mc{L}_\msf{EL}$ with announcement formulas.

An \emph{epistemic frame} $\mc{F}$  is a tuple $(W,\{\sim_a\}_{a\in \msf{Ag}})$, where $W$ is a set of states and $\sim_a \subseteq W\times W$ is a reflexive, transitive and symmetric relation for each $a\in \msf{Ag}$.

An \emph{epistemic model} $\mc{M}$ is a tuple $(W,\{\sim_a\}_{a\in \msf{Ag}},V)$ where $(W,\{\sim_a\}_{a\in \msf{Ag}})$ is an epistemic frame and $V$ is a function from $\msf{Prop}$ to $\mc{P}(W)$.

Let $\mc{M}=(W,\{\sim_a\}_{a\in \msf{Ag}},V)$ be an epistemic model and $w \in W$. The notion of $\varphi$ being \emph{true} at $w$ in $\mc{M}$ (notation: $\mc{M},w \Vdash \varphi$) is defined inductively as follows:
 \[
\begin{array}{lll}
\mathcal{M},w \Vdash p &~\text{iff}~ & w\in V(p)
\\
\mathcal{M},w \Vdash \neg \varphi      &~\text{iff}~& \mathcal{M},w \nVdash \varphi
\\
\mathcal{M},w \Vdash \varphi \land \psi   &  ~\text{iff}~          & \mathcal{M},w \Vdash \varphi        ~\text{and}~         \mathcal{M},w \Vdash \psi
\\
\mathcal{M},w \Vdash \varphi \to \psi   &  ~\text{iff}~          & \mathcal{M},w \nVdash \varphi        ~\text{or}~         \mathcal{M},w \Vdash \psi
\\
\mathcal{M},w \Vdash K_{a}\varphi & ~\text{iff}~ & \text{for~all}~ v  \in W, w\sim_{a}v ~\text{implies}~ \mathcal{M},v \Vdash \varphi
\\
\mathcal{M},w \Vdash \p{\varphi}\psi  &~\text{iff}~  & \mathcal{M},w \Vdash \varphi  ~\text{implies}~ \mathcal{M}^{\varphi},w \Vdash \psi
\end{array}
\]
where $\mathcal{M}^{\varphi}= (W^\varphi, \{\sim^\varphi_a\}_{a\in \msf{Ag}},V^\varphi)$ is called the \emph{model restricted to $\varphi$} where 
$W^\varphi = \{w \in W\mid \mc{M}, w \Vdash \varphi\}$, 
$\sim^\varphi_a = \sim_a \cap (W_\varphi \times W_\varphi)$ and $V^\varphi= V(p) \cap W^\varphi$.

A formula $\varphi$ is \emph{globally true} in an epistemic model (notation: $\mc{M} \Vdash \varphi$) if $\mc{M},w\Vdash \varphi$ for all $w \in W$. A formula $\varphi$ is \emph{valid} in an epistemic frame $\mc{F}$ if $\mc{F}, V\Vdash\varphi$ for all valuations $V$.

Epistemic logic $\msf{EL}$ is the set of $\mc{L}_\msf{EL}$-formulas that are valid on the class of all epistemic frames. Public announcement logic $\msf{PAL}$ is the set of $\mc{L}_\msf{PAL}$-formulas that are valid on the class of all epistemic frames.

$\msf{EL}$ is axiomatized by (Tau), (K), (4), (T), (5), (MP) and (GK$_a$). $\msf{PAL}$ is axiomatized by the axiomatization for $\msf{EL}$ plus \emph{reduction axioms}  (R1--6):
\begin{center}
\begin{tabular}{cc}
\begin{tabular}{rl}
(Tau) & Classical propositional tautologies.\\
(K) & $K_a(\varphi\to\psi)\to(K_a\varphi \to K_a\psi)$\\
(4) & $K_a\varphi\to K_aK_a\varphi$\\
(T) & $K_a\varphi\to \varphi$\\
(5) & $\neg K_a\varphi \to K_a\neg K_a\varphi$\\
(MP) & From $\varphi$ and $\varphi\to\psi$ infer $\psi$.\\
(GK$_a$) & From $\varphi$ infer $K_a\varphi$.\\
\end{tabular}
&
\begin{tabular}{rl}
(R1) & $\p{\varphi}p\leftrightarrow (\varphi\to p)$\\
(R2) & $\p{\varphi}\neg \psi \leftrightarrow (\varphi\to \neg \p{\varphi} \psi)$\\
(R3) & $\p{\varphi}(\psi\wedge\chi)\leftrightarrow (\p{\varphi}\psi\wedge\p{\varphi}\chi)$\\
(R4) & $\p{\varphi}(\psi\to\chi)\leftrightarrow(\p{\varphi}\psi\to\p{\varphi}\chi)$\\
(R5) & $\p{\varphi}K_a\psi \leftrightarrow (\varphi\to K_a\p{\varphi}\psi)$\\
(R6) & $\p{\varphi}\p{\psi}\chi\leftrightarrow \p{\varphi\wedge\p{\varphi}\psi}\chi$\\
\end{tabular}
\end{tabular}
\end{center}

\begin{rem}
The standard language for $\msf{PAL}$ does not contain $\to$. To simplify our writing, we add $\to$ to our language. Because of its existence, $\mrm{(R4)}$ is added to the axiomatization.
\end{rem}

\subsection{Labelled sequent calculus} 
A labelled sequent calculus for a logic with Kripke semantics is based on the internalization of Kripke semantics. 

Let $\mc{F}= (W, \{\sim_a\}_{a \in \msf{Ag}})$ be an epistemic frame. A \emph{relational atom} is of the form $x\sim_a y$, where $x,y \in W$ and $a\in \msf{Ag}$. A \emph{labelled formula} is of the form $x :\varphi$, where $x \in W$ and $\varphi$ is an $\mc{L}_\msf{EL}$-formula. We use $\sigma,\delta$ with or without subscript to denote relational atoms or labelled formulas.

A \emph{multiset} is a `set with multiplicity', or put the other way round,  a sequence modulo the ordering. A \emph{labelled sequent} is of the form $\Gamma \Imp \Delta$ where $\Gamma,\Delta$ are finite multisets of relational atoms and labelled formulas. A \emph{sequent rule} is of the form 
\[
(\mc{R})~\frac{\Gamma_1\Imp \Delta_1~\ldots~ \Gamma_m\Imp \Delta_m}{\Gamma \Imp \Delta}
\]
where $m\ge 0$. $\Gamma_1\Imp \Delta_1,\ldots ,\Gamma_m\Imp \Delta_m$ are called \emph{premises} of this rule and $\Gamma \Imp \Delta$ is called the conclusion. If $m=0$, we simply write $\Gamma \Imp \Delta$ and call it an \emph{initial sequent}. 
The formula with the connective in a rule is the \emph{principal} formula of that rule, and its components in the premisses are the \emph{active} formulas.
A \emph{labelled sequent calculus} is a set of sequent rules. A derivation in a labelled sequent calculus $G$ is defined as usual.  The \emph{derivation height h} of a sequent is defined as the length of longest branch in the derivation of the sequent. 
We use $G\vdash \Gamma \Imp \Delta$ to denote that $\Gamma \Imp \Delta$ is derivable in $G$ and $G\vdash_h \Gamma \Imp \Delta$ to denote that $\Gamma \Imp \Delta$ is derivable in $G$ with a derivation the height of which is at most $h$.

A sequent rule $(\mc{R})$ 
is \emph{admissible} in $G$ if $G\vdash \Gamma_i \Imp \Delta_i$ for $1\le i\le m$ implies $G\vdash \Gamma \Imp \Delta$. It is \emph{height-preserving admissible} if $G\vdash_h \Gamma_i \Imp \Delta_i$ for $1\le i\le m$ implies $G\vdash_h \Gamma \Imp \Delta$.

Let $\mc{M}= (W, \{\sim_a\}_{a \in \msf{Ag}},V)$ be an epistemic model. The \emph{interpretation} $\tau_\mc{M}$ of relational atoms and labelled formulas are defined as follows:
\[
\begin{aligned}
\tau_\mc{M} (x \sim_a y)&\quad=\quad &x \sim_a y;\\
\tau_\mc{M}(x: \varphi) &\quad= \quad& \mc{M},x \Vdash \varphi.
\end{aligned}
\]

A labelled sequent $\sigma_1,\ldots, \sigma_m \Imp \delta_1,\ldots, \delta_n$ is \emph{valid} if 
\[
 \forall \mc{M} \forall x_1\ldots \forall x_k[\tau_\mc{M}(\sigma_1) \land \ldots \land \tau_\mc{M}(\sigma_m)\to \tau_\mc{M}(\delta_1)\lor\ldots\lor\tau_\mc{M}(\delta_n)]
\]
is true, where $\mc{M}=(W,\{\sim_a\}_{a\in \msf{Ag}},V)$ is an epistemic model, and $x_1,\ldots,x_k$ are variables occurring in $\sigma_1,\ldots, \sigma_m \Imp \delta_1,\ldots, \delta_n$ ranging over $W$.

 A sequent rule $\mc{R}$
is \emph{valid} if the validity of all the premises implies the validity of the conclusion.

Given a labelled sequent calculus $G$ and a logic $\Lambda$, we say $G$ is a \emph{labelled sequent calculus for $\Lambda$} if for all $\varphi$, $G\vdash{} \Imp \varphi$ if and only if $\varphi\in \Lambda$.

\subsection{Labelled sequent calculus for $\msf{EL}$}

\begin{defn}
Labelled sequent calculus $G_\msf{EL}$ for $\msf{EL}$ consists of the following initial sequents and rules\footnote{It is mentioned without proof in \cite{hakli2007proof} that this is a labelled sequent calculus for $\msf{EL}$. This calculus is a multi-agent version of the labelled sequent calculus for $\msf{S}5$ proposed in \cite{negri2005proof}.}:

(1) Initial sequents:
\begin{align*}
x:p,\Gamma \Imp \Delta,x:p
&\quad&
x\sim_ay,\Gamma \Imp \Delta, x\sim_ay
\end{align*}

(2) Propositional rules:
\begin{align*}
&(\neg{\Imp})~\frac{\Gamma \Imp \Delta,x:\varphi}{x:\neg \varphi,\Gamma \Imp \Delta}
&
&({\Imp}\neg)~\frac{x:\varphi,\Gamma \Imp\Delta}{\Gamma \Imp \Delta,x:\neg \varphi}
\\
&(\land\! \Imp )~\frac{x:\varphi_1, x:\varphi_2,\Gamma \Imp \Delta}{x:\varphi_1\land \varphi_2,\Gamma \Imp \Delta}
&
&(\Imp\!\land)~\frac{\Gamma \Imp \Delta,x:\varphi_1 \quad \Gamma \Imp \Delta, x:\varphi_2}{\Gamma \Imp \Delta, x:\varphi_1\land \varphi_2}
\\
&(\to\Imp)~\frac{\Gamma \Imp \Delta, x:\varphi\quad x:\psi,\Gamma \Imp \Delta}{x:\varphi \to \psi, \Gamma \Imp \Delta}
&
&(\Imp\to)~\frac{x:\varphi ,\Gamma \Imp \Delta, x:\psi}{\Gamma \Imp \Delta, x:\varphi \to \psi}
\end{align*}

(3) Modal rules:
\begin{align*}
&(K_{a} \Imp )~\frac{y :\varphi, x: K_{a} \varphi,x\sim_ay, \Gamma \Imp \Delta}{x:K_{a} \varphi,x\sim_ay, \Gamma \Imp \Delta}
&
&
(\Imp K_{a})~\frac{x\sim_ay,\Gamma \Imp \Delta, y:\varphi}{\Gamma \Imp \Delta, x:K_{a} \varphi}
\end{align*}
where $y$ does not occur in the conclusion of  $(\Imp K_{a})$, .

(4) Relational rules:

\begin{align*}
&(\text{Ref}_a)~\frac{x\sim_ax,\Gamma \Imp \Delta}{\Gamma \Imp \Delta}
&
&(\text{Trans}_a)~\frac{x\sim_az,x\sim_ay,y\sim_az,\Gamma \Imp \Delta}{x\sim_ay,y\sim_az,\Gamma \Imp \Delta}
\\
&(\text{Sym}_a)~\frac{y\sim_ax,x\sim_ay,\Gamma \Imp \Delta}{x\sim_ay,\Gamma \Imp \Delta}
\end{align*}

\end{defn}

\begin{prop}\label{prop:iden}
For any $\mc{L}_\msf{EL}$-formula $\varphi$, $G_\msf{EL} \vdash x:\varphi, \Gamma \Imp \Delta,x:\varphi$.
\end{prop}

\begin{proof}
It can be proved by induction on $\varphi$.
\end{proof}

By proofs similar to those in \cite{negri2005proof}, we have Theorems \ref{thm:adm} and \ref{thm:ElDeci}.

\begin{thm}\label{thm:adm}
Structural rules $(w\Imp)$, $(\Imp w)$, $(c\Imp)$, $(\Imp c)$, $(c_R\Imp)$ and $(\Imp c_R)$ are height-preserving admissible in $G_\msf{EL}$. The cut rule $(Cut)$ is admissible in $G_\msf{EL}$.

\begin{align*}
&(w\Imp)~\frac{\Gamma \Imp \Delta}{x:\varphi, \Gamma \Imp \Delta}&
\quad
&(\Imp w)~\frac{\Gamma \Imp \Delta}{\Gamma \Imp \Delta, x:\varphi}
\\
&(c\Imp)~\frac{x:\varphi,x:\varphi,\Gamma \Imp \Delta}{x:\varphi,\Gamma \Imp \Delta}
&
&(\Imp c)~\frac{\Gamma \Imp \Delta,x:\varphi,x:\varphi}{\Gamma \Imp \Delta, x:\varphi}
\\
&(c_R\Imp)~\frac{x\sim_ay,x\sim_ay,\Gamma \Imp \Delta}{x\sim_ay,\Gamma \Imp \Delta}
&
&(\Imp c_R)~\frac{\Gamma \Imp \Delta,x\sim_ay,x\sim_ay}{\Gamma \Imp \Delta, x\sim_ay}
\\
&(Cut)~\frac{\Gamma \Imp \Delta, x:\varphi \quad x:\varphi ,\Gamma' \Imp \Delta'}{\Gamma ,\Gamma'\Imp \Delta, \Delta'}
\end{align*}
\end{thm}

\begin{thm}\label{thm:ElDeci}
$G_\msf{EL}$ allows terminating proof search.
\end{thm}

\begin{exam}
A derivation for axiom $(5)$ in $G_\msf{EL}$ is as follows:
\[
\AxiomC{$z:\varphi, y: K_a\varphi, y\sim_a z, x \sim_a z, y\sim_a x, x \sim_a y \Imp z: \varphi$}
\ri{\f$( K_a\Imp)$}
\UnaryInfC{$ y: K_a\varphi, y\sim_a z, x \sim_a z, y\sim_a x, x \sim_a y \Imp z: \varphi$}
\ri{\f$( Tran_a)$}
\UnaryInfC{$  y: K_a\varphi, x \sim_a z, y\sim_a x, x \sim_a y \Imp z: \varphi$}
\ri{\f$( Sym_a)$}
\UnaryInfC{$  y: K_a\varphi, x \sim_a z, x \sim_a y \Imp z: \varphi$}
\ri{\f$( \Imp \neg)$}
\UnaryInfC{$  x \sim_a z, x \sim_a y \Imp y:\neg K_a\varphi,z: \varphi$}
\ri{\f$( \Imp K_a)$}
\UnaryInfC{$  x \sim_a y \Imp y:\neg K_a\varphi,x: K_a\varphi$}
\ri{\f$( \Imp K_a)$}
\UnaryInfC{$  \Imp x:K_a\neg K_a\varphi,x: K_a\varphi$}
\ri{\f$(\neg \Imp)$}
\UnaryInfC{$x:\neg K_a\varphi  \Imp x:K_a\neg K_a\varphi$}
\ri{\f$(\Imp \to)$}
\UnaryInfC{$\Imp x:\neg K_a\varphi \to K_a\neg K_a\varphi$}
\DisplayProof
\]
\end{exam}

\section{Labelled sequent calculus for $\msf{PAL}$}

In this section, we introduce a labelled sequent calculus for $\msf{PAL}$. The Hilbert-style axiomatization for $\msf{PAL}$ is the extension of that for $\msf{EL}$ with reduction axioms. Can we obtain a labelled sequent calculus for $\msf{PAL}$ by adding some rules adapted from reduction axioms to the labelled sequent calculus $G_\msf{EL}$ for $\msf{EL}$? The answer is yes and this is what we do.

Take reduction axiom (R1)$\p{\varphi}p\leftrightarrow (\varphi\to p)$ as an example. The equivalence symbol in the axiom means that $\p{\varphi}p$ and $\varphi\to p$ are equivalent in $\msf{PAL}$. Therefore, the most direct sequent rules for (R1) are 
\[
\frac{x:\varphi\to p,\Gamma \Imp \Delta}{x:\p{\varphi}p, \Gamma \Imp \Delta} \text{~~~and~~~} \frac{\Gamma \Imp \Delta, x:\varphi\to p}{\Gamma \Imp \Delta,x: \p {\varphi}p}
\]
The rule on the left is sound because of $\p{\varphi} p\to (\varphi\to p)$ and the rule on the right is sound because of $ (\varphi\to p) \to \p{\varphi}p $. These rules can be written in a more neat way if we see $\varphi\to p,\Gamma \Imp \Delta$ as the conclusion of an application of $(\to \Imp)$ and $\Gamma \Imp \Delta, \varphi\to p$ as the conclusion of an application of $(\Imp \to)$. Applying the reverse of $(\to \Imp)$ and $(\Imp \to)$ to their premises, we have the rules for (R1) that will be added to $G_\msf{EL}$:
\[
\frac{\Gamma \Imp \Delta, x:\varphi \quad x:p,\Gamma \Imp \Delta}{x:\p{\varphi}p, \Gamma \Imp \Delta} \text{~~~and~~~} \frac{x:\varphi, \Gamma \Imp \Delta,  p}{\Gamma \Imp \Delta, x:\p {\varphi}p}
\]
In a similar way, we can have sequent rules for other reduction axioms. Therefore, we have:
\begin{defn}\label{defn:calculus}
Labelled sequent calculus $G_\msf{PAL}$ for $\msf{PAL}$ is $G_\msf{EL}$ plus the following sequent rules:
\begin{align*}
&\small{(R1\!\Imp)~\frac{\Gamma \Imp \Delta, x:\varphi \quad x:p ,\Gamma \Imp \Delta}{ x:\p{\varphi} p,\Gamma \Imp \Delta}}
&
\quad
&\small{(\Imp\!R1)~\frac{x:\varphi,\Gamma \Imp \Delta, x:p}{\Gamma \Imp \Delta, x:\p{\varphi}p}}
\\
&(R2\Imp)~\frac{\Gamma \Imp \Delta, x: \varphi\quad x:\neg\p{\varphi}\psi,\Gamma \Imp \Delta}{x: \p{\varphi}\neg \psi ,\Gamma \Imp \Delta}&
\quad
&(\Imp R2)~\frac{x:\varphi, \Gamma \Imp \Delta, x:\neg\p{\varphi}\psi}{\Gamma \Imp \Delta, x:\p{\varphi}\neg \psi}
\\
&(R3 \Imp)~\frac{x:\p{\varphi}\psi_1,x:\p{\varphi}\psi_2,\Gamma \Imp \Delta}{x:\p{\varphi}(\psi_1\land \psi_2),\Gamma \Imp \Delta}&
\quad
&(\Imp R3)~\frac{\Gamma \Imp \Delta, x:\p{\varphi}\psi_1  \quad \Gamma \Imp \Delta, x:\p{\varphi}\psi_2}{\Gamma \Imp \Delta, x:\p{\varphi}(\psi_1\land \psi_2)}
\\
&(R4\Imp)~\frac{\Gamma \Imp \Delta, x:\p{\varphi} \psi_1 \quad x:\p{\varphi}\psi_2,\Gamma \Imp \Delta}{x:\p{\varphi}(\psi_1\to \psi_2),\Gamma \Imp \Delta}&
\quad
&(\Imp R4)~\frac{x:\p{\varphi}\psi_1,\Gamma \Imp \Delta,x:\p{\varphi}\psi_2}{\Gamma \Imp \Delta, \p{\varphi}(\psi_1\to\psi_2)}
\\
&(R5\Imp)~\frac{\Gamma \Imp \Delta, x: \varphi\quad K_a\p{\varphi}\psi,\Gamma \Imp \Delta}{x: \p{\varphi}K_a\psi ,\Gamma \Imp \Delta}&
\quad
&(\Imp R5)~\frac{x:\varphi, \Gamma \Imp \Delta, x:K_a\p{\varphi}\psi}{\Gamma \Imp \Delta, x:\p{\varphi}K_a\psi}
\\
&(R6\Imp)~\frac{x:\p{\varphi\wedge\p{\varphi}\psi}\chi,\Gamma \Imp \Delta}{x:\p{\varphi}\p{\psi}\chi,\Gamma \Imp \Delta}&
\quad
&(\Imp R6)\frac{\Gamma \Imp \Delta,x:\p{\varphi\wedge\p{\varphi}\psi}\chi}{\Gamma \Imp \Delta,x:\p{\varphi}\p{\psi}\chi} 
\end{align*}
We call these sequent rules the \emph{reduction rules}.
\end{defn}

There are six pairs of reduction rules in $G_\msf{PAL}$, each pair dealing with a kind of announcement formulas. Each left rule introduces a formula on the left of $\Imp$, and each right rule introduces one on the right of $\Imp$. Another desirable property for sequent rules is that the complexity of each premise should be less than that of the conclusion.
If we define the complexity of a sequent to be the the sum of relational atoms and labelled formulas occurring in it, then the definition of formula complexity that counts the number of connectives will make (R5) and (R6) fail to satisfy the complexity increasing property. The following definition for formula complexity can solve this problem\footnote{This is Definition 7.21 in \cite{DIT2007}.}:

\begin{defn}\label{defn:compl}
Let $\varphi$ be an $\mc{L}_\msf{PAL}$ formula, The \emph{complexity} $c(\varphi)$ of $\varphi$ is defined as follows:
\begin{center}
\begin{tabular}{ll}
$c(p)=1$ & $c(\varphi \to \psi)= 1+ \max{\{c(\varphi),c(\psi)\}}$    \\
$c(\neg \varphi)=1+c(\varphi)$ & $c(K_a \varphi) = 1 +c(\varphi)$\\
$c(\varphi \land\psi)= 1+ \max{\{c(\varphi),c(\psi)\}}$ &$c(\p{\varphi}\psi)= (4+c(\varphi) )\cdot c(\psi)$.  \\
\end{tabular}
\end{center}
\end{defn}

Then we have the following lemma\footnote{This is Lemma 7.22 in \cite{DIT2007}.}:
\begin{lem}
For all $\mc{L}_\msf{PAL}$-formulas $\varphi,\psi$ and $\chi$:

\begin{minipage}[t]{0.45\textwidth}
\begin{enumerate}
\item[(1)] $c(\p{\varphi} p)>  c(\varphi \to p)$;
\item[(2)] $c(\p{\varphi}\neg \psi )>c(\varphi\to \neg\p{\varphi} \psi)$;
\item[(3)] $c(\p{\varphi}(\psi \land \chi)) > c(\p{\varphi}\psi \land \p{\varphi}\chi)$;
\end{enumerate}
\end{minipage}
\begin{minipage}[t]{0.45\textwidth}
\begin{enumerate}
\item[(4)] $c(\p{\varphi}K_a \psi)> c(\varphi \to K_a \p{\varphi} \psi)$;
\item[(5)] $c(\p{\varphi}\p{\psi}\chi)>c (\p{\varphi \land \p{\varphi}\psi}\chi)$.
\end{enumerate}
\end{minipage}

\end{lem}

\begin{lem}
For any $\mc{L}_\msf{PAL}$-formula $\varphi$, $G_\msf{PAL}\vdash x:\varphi,\Gamma \Imp \Delta, x:\varphi$.

\end{lem}

\begin{proof}
\begin{proof}
We prove this by induction on the structure of $\varphi$ with a subinduction on $\psi$ of the inductive case where $\varphi$ equals $[\varphi]\psi$. All inductive case not involving announcement are the same as in the proof of Proposition \ref{prop:iden}.  
When $\varphi$ involves a public announcement operator, there are 6 subcases. We show two representative cases. 

When $\varphi=\p{\phi}K_a\psi$, the derivation is as follows:
\[
\AxiomC{$x:\phi, \Gamma\Imp\Delta, x: K_a\p{\phi}\psi, x:\phi$}
\AxiomC{$x:K_a[\phi]\psi, x:\phi, \Gamma\Imp\Delta,x: K_a\p{\phi}\psi$}
\ri{\f$(R5\Imp)$}
\BinaryInfC{$x:\phi, x:\p{\phi}K_a\psi, \Gamma\Imp\Delta, x: K_a\p{\phi}\psi$}
\ri{\f$(\Imp R5)$}
\UnaryInfC{$x:\p{\phi}K_a\psi, \Gamma\Imp\Delta, x:\p{\phi}K_a\psi$}
\DisplayProof
\]

When $\varphi=\p{\phi}[\chi]\psi$, the derivation is as follows:
\[
\AxiomC{$x:[\phi\land[\chi]\psi], \Gamma\Imp\Delta, x:[\phi\land[\phi]\chi]\psi$}
\ri{\f$(R6\Imp)$}
\UnaryInfC{$x: \p{\phi}[\chi]\psi,\Gamma\Imp\Delta, x:[\phi\land[\phi]\chi]\psi$}
\ri{\f$(\Imp R6)$}
\UnaryInfC{$x: \p{\phi}[\chi]\psi,\Gamma\Imp\Delta, x:\p{\phi}[\chi]\psi$}
\DisplayProof
\]
Other cases can be proved analogously. 
\end{proof}

\end{proof}

\begin{exam}\label{exa:0414}
Now we show that $(R5)\p{\varphi}K_a \psi \leftrightarrow (\varphi \to K_a \p{\varphi}\psi)$ is derivable in $G_\msf{PAL}$. A derivation for $\p{\varphi}K_a \psi \to (\varphi \to K_a \p{\varphi}\psi)$ in $G_\msf{PAL}$ is as follows:
\[
\AxiomC{$x\sim_a y , x: \varphi \Imp y: \p{\varphi}\psi,x:\varphi$}
\AxiomC{$y: \p{\varphi} \psi, x:K_a \p{\varphi} \psi, x\sim_a y , x: \varphi \Imp y: \p{\varphi}\psi$}
\ri{\f$(K_a\! \Imp)$}
\UnaryInfC{$x:K_a \p{\varphi} \psi, x\sim_a y , x: \varphi \Imp y: \p{\varphi}\psi$}
\ri{\f$(R5\Imp)$}
\BinaryInfC{$x\sim_a y ,x:\p{\varphi}K_a \psi,  x: \varphi \Imp y: \p{\varphi}\psi$}
\ri{\f$(\Imp K_a)$}
\UnaryInfC{$x:\p{\varphi}K_a \psi,  x: \varphi \Imp x: K_a \p{\varphi}\psi$}
\ri{\f$(\Imp \to)$}
\UnaryInfC{$x:\p{\varphi}K_a \psi \Imp x: \varphi \to K_a \p{\varphi}\psi$}
\ri{\f$(\Imp \to)$}
\UnaryInfC{$\Imp x:\p{\varphi}K_a \psi \to (\varphi \to K_a \p{\varphi}\psi)$}
\DisplayProof
\]
A derivation for $(\varphi \to K_a \p{\varphi}\psi) \to \p{\varphi}K_a \psi$ in $G_\msf{PAL}$ is as follows:
\[
\AxiomC{$x:\varphi \Imp x:K_a  \p{\varphi}\psi,x:\varphi$}
\AxiomC{$ x:K_a  \p{\varphi}\psi, x:\varphi \Imp x:K_a  \p{\varphi}\psi $}
\ri{\f$(\to \Imp)$}
\BinaryInfC{$x:\varphi , x:  \varphi\to  K_a \p{\varphi}\psi \Imp x:K_a  \p{\varphi}\psi$}
\ri{\f$(\Imp R5)$}
\UnaryInfC{$x:  \varphi\to  K_a \p{\varphi}\psi \Imp x: \p{\varphi}K_a \psi$}
\ri{\f$(\Imp \to)$}
\UnaryInfC{$\Imp x: (\varphi \to K_a \p{\varphi}\psi) \to \p{\varphi}K_a \psi)$}
\DisplayProof
\]
\end{exam}

\begin{example}
$[p\land\neg K_ap]\neg K_ap$ is not derivable in $G_{\msf{PAL}}$.
\[
\small
\AxiomC{$\mc{D}_0$}
\UnaryInfC{$x\sim_ay,x:K_a[p\land\neg K_ap]p,x:p\Imp y:p$}
\ri{\f$(\Imp K_a)$}
\UnaryInfC{$x:K_a[p\land\neg K_ap]p,x:p\Imp x:K_ap$}
\AxiomC{$x:p\Imp x:p\land\neg K_ap,x:K_ap$}
\BinaryInfC{$x:[p\land\neg K_ap]K_ap, x:p\Imp x:K_ap$}
\ri{\f$(\neg\Imp)$}
\UnaryInfC{$x:[p\land\neg K_ap]K_ap, x:p, x:\neg K_ap\Imp$}
\ri{\f$(\land\Imp,\Imp\neg)$}
\UnaryInfC{$x:p\land\neg K_ap\Imp x:\neg [p\land\neg K_ap]K_ap$}
\UnaryInfC{$\Imp x:[p\land\neg K_ap]\neg K_ap$}
\DisplayProof
\]
where $\mc{D}_0$ is:
\[
\small
\AxiomC{$x\sim_ay,x\!:\!p\Imp y\!:\!p,y\!:\!p$}
\AxiomC{$y\!:\!K_ap,x\sim_ay,x\!:\!p\Imp y\!:\!p$}
\UnaryInfC{$x\sim_ay,x\!:\!p\Imp y\!:\!p, y\!:\!\neg K_ap$}
\BinaryInfC{$x\sim_ay,x\!:\!p\Imp y\!:\!p,y\!:\!p\land\neg K_ap$}
\AxiomC{$y\!:\!p,x\sim_ay,x\!:\!p\Imp y\!:\!p$}
\BinaryInfC{$y\!:\![p\land\neg K_ap]p,x\sim_ay,x\!:\!p\Imp y\!:\!p$}
\UnaryInfC{$x\sim_ay,x\!:\!K_a[p\land\neg K_ap]p,x\!:\!p\Imp y\!:\!p$}
\DisplayProof
\]
\end{example}

\section{Admissibility of some structural rules}
In light of the reduction axioms, we can define a translation from $\mc{L}_\msf{PAL}$-formulas to $\mc{L}_\msf{EL}$-formulas\footnote{This is Definition 7.20 in \cite{DIT2007}.}.
\begin{defn}
The translation $t: \mc{L}_\msf{PAL}\to \mc{L}_\msf{EL}$ is defined as follows:

\begin{center}
\begin{tabular}{c c}
   {$   
      \begin{aligned}
&t(p)  &=\quad &p \\
&t(\neg \varphi) &=\quad &\neg t (\varphi)\\
&t(\varphi \land \psi) &=\quad& t(\varphi) \land t (\psi)\\
&t(\varphi \to \psi) &=\quad& t(\varphi) \to t (\psi)\\
&t(K_a\varphi) &=\quad& K_a t(\varphi)\\

\end{aligned}
     $}  & {$\begin{aligned}
     &t(\p{\varphi}p) &=\quad& t(\varphi \to p)\\
&t(\p{\varphi}\neg \psi) &=\quad& t (\varphi \to \neg \p{\varphi} \psi)\\
&t(\p{\varphi}(\psi \land \chi)) &= \quad&t(\p{\varphi} \psi \land \p{\varphi}\chi)\\
&t(\p{\varphi}(\psi \to \chi) )&= \quad&t(\p{\varphi} \psi \to \p{\varphi}\chi)\\
&t(\p{\varphi}K_a \psi) &=\quad&t (\varphi \to K_a\p{\varphi} \psi)\\
&t (\p{\varphi}\p{\psi}\chi) &=\quad &t (\p{\varphi \land \p{\varphi}\psi} \chi)
     
   \end{aligned}$}\\

\end{tabular}
\end{center}

\end{defn}
Now we extend translation $t$ to relational atoms and labelled $\mc{L}_\msf{PAL}$-formulas: for any relational atom $x\sim_ay$, let $t(x\sim_ay) =x\sim_ay$; for any labelled $\mc{L}_\msf{PAL}$-formula $x:\varphi$, $t(x:\varphi) = x:t(\varphi)$. Moreover, for any set $\Gamma$ of relational atoms and labelled formulas: $t(\Gamma)= \{t(\sigma)\mid \sigma \in \Gamma\}$.

\begin{lem}\label{lem:07081609}
For any $\mc{L}_\msf{PAL}$-sequent $x:\varphi,\Gamma \Imp \Delta$, the following hold:
\begin{enumerate}
\item  if $G_\msf{PAL} \vdash  x:t(\varphi), t(\Gamma)\Imp t(\Delta)$, then $G_\msf{PAL}\vdash x:\varphi, t(\Gamma) \Imp t(\Delta)$;
\item  if $G_\msf{PAL} \vdash t(\Gamma)\Imp t(\Delta), x:t(\varphi)$, then $G_\msf{PAL}\vdash t(\Gamma) \Imp t(\Delta), x:\varphi$.
\end{enumerate}
\end{lem}

\begin{proof}
We prove these claims simultaneously by induction on the height of derivation $h$ of $G_{\msf{PAL}}\vdash x:t(\varphi),t(\Gamma)\Imp t(\Delta)$ (or $G_{\msf{PAL}}\vdash t(\Gamma)\Imp t(\Delta),x:t(\varphi)$).

If $h=1$, then $x:t(\varphi), t(\Gamma)\Imp t(\Delta)$ (or $t(\Gamma)\Imp t(\Delta),x:t(\varphi)$) is an initial sequent. If $x:t(\varphi)$ is principal, then $t(\varphi)=p$ for some proposition letter $p$. It follows that $\varphi =p$. Therefore, $x:\varphi, t(\Gamma)\Imp t(\Delta)$ (or $t(\Gamma)\Imp t(\Delta),x:\varphi$) is also an initial sequent. If $x:t(\varphi)$ is not principal, it is immediate that $x:\varphi, t(\Gamma)\Imp t(\Delta)$ (or $t(\Gamma)\Imp t(\Delta), x:\varphi$) is an initial sequent. 

If $h>1$, the induction hypothesis is formulated as:

\noindent
$(1)'$for all $i<h$, if  $G_\msf{PAL} \vdash_i  x:t(\varphi), t(\Gamma)\Imp t(\Delta)$, then $G_\msf{PAL}\vdash x:\varphi, t(\Gamma) \Imp t(\Delta)$;

\noindent
$(2)'$for all $i<h$, if  $G_\msf{PAL} \vdash_i   t(\Gamma)\Imp t(\Delta),x:t(\varphi)$, then $G_\msf{PAL}\vdash  t(\Gamma) \Imp t(\Delta),x:\varphi$.

In what follows we only give the proof for claim $(1)'$. The proof for the other claim is similar.

Assume that $G_\msf{PAL} \vdash_h  x:t(\varphi), t(\Gamma)\Imp t(\Delta)$. Then there exists a derivation $\mc{D}$ for $x:t(\varphi), t(\Gamma)\Imp t(\Delta)$ in $\GP$. Let the last rule applied in $d$ be $\mc{R}$. If $x:t(\varphi)$ is not principal in the application of $\mc{R}$, the desired result can be obtained by applying the induction hypothesis to the premise(s) of $\mc{R}$ and then applying $\mc{R}$.

If $x:t(\varphi)$ is principal in the application of $\mc{R}$, we prove by an sub-induction on the complexity $c(\varphi)$ of $\varphi$. Since $x:t(\varphi)$ is principal and $h>1$, $\varphi$ is not a proposition letter. We have ten sub-cases. We divide them into two groups depending on whether $\varphi$ starts with an announcement operator or not.

If $\varphi$ does not start with an announcement operator, the desired result can be obtained by applying the induction hypothesis to the premise(s) of $\mc{R}$ and then applying $\mc{R}$. There are four sub-cases: $\varphi$ is of the form $\neg \psi$, $\psi_1\land \psi_2$, $\psi_1\to \psi_2$ or $K_a\psi$. We illustrate this by the cases $\neg \psi$ and $K_a \psi$.

\begin{itemize}
\item If $\varphi = \neg \psi$, then $\mc{R}$ is $(\neg\!\Imp)$. Note that $t (\varphi) = t(\neg \psi) = \neg t(\psi)$. Let the derivation $\mc{D}$ end with
 \[
 \AxiomC{$ t(\Gamma)\Imp t(\Delta), x:t(\psi)$}
\ri{\f$(\neg\! \Imp )$}
\UnaryInfC{$x:\neg t(\psi) , t(\Gamma)\Imp t(\Delta)$}
\DisplayProof
\]
By the induction hypothesis, we have $\GP \vdash t(\Gamma)\Imp t(\Delta), x:\psi$. Then by $(\neg\! \Imp)$ we have $\GP \vdash x:\neg \psi,t(\Gamma)\Imp t(\Delta)$.

\item If $\varphi =K_a \psi$,  then $\mc{R}$ is $(K_a \Imp)$. Note that $t(\varphi) =t(K_a\psi) =K_a t(\psi)$. Let the derivation $\mc{D}$ end with 
\[
\AxiomC{$y:t(\psi),x:K_at(\psi),x \sim_a y, t(\Gamma) \Imp t(\Delta)$}
\ri{\f$(K_a \Imp )$}
\UnaryInfC{$x:K_at(\psi),x \sim_a y, t(\Gamma) \Imp t(\Delta)$}
\DisplayProof
\]
First apply the main induction hypothesis to $x:K_at(\psi)$ and we have $\GP\vdash y:t(\psi),x:K_a \psi,x \sim_a y, t(\Gamma) \Imp t(\Delta)$. Then apply the sub-induction hypothesis to $y:t(\psi)$ and we have $\GP \vdash y:\psi,x:K_a \psi,x \sim_a y, t(\Gamma) \Imp t(\Delta)$. Finally by $(K_a \Imp )$ we have $\GP\vdash x:K_a \psi,x \sim_a y, t(\Gamma) \Imp t(\Delta)$.
\end{itemize}

If $\varphi$ starts with an announcement operator, then there are six sub-cases: $\varphi$ is $\p{\phi}p$, $\p{\phi}\neg \psi$, $\p{\phi}(\psi_1\land\psi_2)$, $\p{\phi}(\psi_1\to \psi_2)$, $\p{\phi}K_a\psi$ or $\p{\phi}\p{\psi}\chi$. 

If $\varphi$ is $\p{\phi}p$, $\p{\phi}\neg \psi$, $\p{\phi}(\psi_1\to \psi_2)$ or $\p{\phi}K_a\psi$, then $t(\varphi)$ is $t(\phi) \to p$, $t(\phi) \to t(\neg \p{\phi}\psi)$, $t(\p{\phi}\psi_1)\to t(\p{\phi}\psi_2)$, or $t(\phi) \to t(K_a \p{\phi}\psi)$, respectively. Since $x:t(\varphi)$ is principal, $\mc{R}$ must be $(\to \Imp)$.We substitute the application of $(\to \Imp)$ with  an application of $(R1\Imp)$, $(R_2\Imp)$, $(R_4\Imp)$ and $(R_5\Imp)$, respectively. We illustrate the proof by the case where $\varphi$ is $\p{\phi}\neg \psi$ and the case where $\varphi$ is $\p{\phi}K_a\psi$.

\begin{itemize}
\item If $\varphi$ is $\p{\phi}\neg \psi$, then the derivation $\mc{D}$ ends with
\[
\AxiomC{$t(\Gamma)\Imp t(\Delta), x:t(\phi)$}
\AxiomC{$x:t(\neg \p{\phi}\psi) ,t(\Gamma) \Imp t(\Delta)$}
\ri{\f$(\to \Imp)$}
\BinaryInfC{$x:t(\phi) \to t(\neg \p{\phi}\psi),t(\Gamma)\Imp t(\Delta)$}
\DisplayProof
\]
Apply the induction hypothesis to the premises and we have $\GP\vdash t(\Gamma)\Imp t(\Delta), \phi$ and $\GP\vdash \neg \p{\phi}\psi ,t(\Gamma) \Imp t(\Delta)$. Then by $(R2\Imp)$ we have $\GP\vdash x:\p{\phi}\neg \psi, t(\Gamma)\Imp t(\Delta)$.

\item If $\varphi$ is $\p{\phi}K_a \psi$, then the derivation $\mc{D}$ ends with
\[
\AxiomC{$t(\Gamma)\Imp t(\Delta), x:t(\phi)$}
\AxiomC{$x:t(K_a \p{\phi}\psi) ,t(\Gamma) \Imp t(\Delta)$}
\ri{\f$(\to \Imp)$}
\BinaryInfC{$x:t(\phi) \to t(K_a \p{\phi}\psi),t(\Gamma)\Imp t(\Delta)$}
\DisplayProof
\]
Apply the induction hypothesis to the premises and we have $\GP\vdash t(\Gamma)\Imp t(\Delta), \phi$ and $\GP\vdash K_a \p{\phi}\psi ,t(\Gamma) \Imp t(\Delta)$. Then by $(R5\Imp)$ we have $\GP\vdash x:\p{\phi}K_a \psi, t(\Gamma)\Imp t(\Delta)$.
\end{itemize}

 If $\varphi$ is $\p{\phi}(\psi_1\land\psi_2)$, then $t(\varphi)$ is $t(\p{\phi}\psi_1)\land t(\p{\phi}\psi_2)$. Since $x:t(\varphi)$ is principal, $\mc{R}$ must be $(\land\!\Imp)$. We substitute the application of $(\land\!\Imp)$ with an application of $(R3\Imp)$. Let the derivation $\mc{D}$ end with 
\[
\AxiomC{$x:t(\p{\phi}\psi_1), x:t(\p{\phi}\psi_2),t(\Gamma)\Imp t(\Delta)$}
\ri{\f$(\land\Imp)$}
\UnaryInfC{$x:t(\p{\phi}\psi_1)\land t(\p{\phi}\psi_2),t(\Gamma)\Imp t(\Delta)$}
\DisplayProof
\]
Apply the sub-induction hypothesis to the premise twice and we have $\GP\vdash x:\p{\phi}\psi_1, x:\p{\phi}\psi_2,t(\Gamma)\Imp t(\Delta)$. Then by $(R3\Imp)$ we have $\GP\vdash x: \p{\phi}(\psi_1\land \psi_2), t(\Gamma)\Imp t(\Delta)$. 

 If $\varphi$ is $\p{\phi}\p{\psi}\chi$, by assumption, $\GP\vdash x: t(\p{\phi}\p{\psi}\chi),t(\Gamma)\Imp t(\Delta)$. Since $t(\p{\phi}\p{\psi}\chi) = t(\p{\phi\land \p{\phi}\psi}\chi)$, $\GP\vdash x:t(\p{\phi\land \p{\phi}\psi}\chi),t(\Gamma)\Imp t(\Delta)$. Since $c(\p{\phi\land \p{\phi}\psi}\chi) <c(\p{\phi}\p{\psi}\chi)$, by the sub-induction hypothesis, we have $\GP\vdash x: \p{\phi\land \p{\phi}\psi}\chi,t(\Gamma)\Imp t(\Delta)$. Then by $(R6\Imp)$, $\GP\vdash x:\p{\phi}\p{\psi}\chi,t(\Gamma)\Imp t(\Delta)$.
 
This completes the proof.
\end{proof}

The following theorem is a bridge between $G_\msf{PAL}$ and $G_\msf{EL}$, enabling us to prove properties of $G_\msf{PAL}$ through $G_\msf{EL}$. 

\begin{thm}\label{thm:10}
For any $\mc{L}_\msf{PAL}$-sequent $\Gamma \Imp \Delta$, 
\begin{enumerate}
\item if $G_\msf{EL}\vdash t(\Gamma) \Imp t(\Delta)$, then $G_\msf{PAL}\vdash \Gamma \Imp \Delta$;
\item if $G_\msf{PAL}\vdash_h \Gamma \Imp \Delta$, then $G_\msf{EL} \vdash_h t(\Gamma)\Imp t(\Delta)$.
\end{enumerate}
\end{thm}

\begin{proof}

(1) Assume that $G_\msf{EL} \vdash t(\Gamma) \Imp t(\Delta)$. Since $G_\msf{PAL}$ is an extension of $G_\msf{EL}$, $G_\msf{PAL} \vdash t(\Gamma) \Imp t(\Delta)$. Since $t(\Gamma) \Imp t(\Delta)$ is finite, applying Lemma \ref{lem:07081609} a finite number of times, we have $G_\msf{PAL} \vdash \Gamma \Imp \Delta$.

(2)  Prove by induction on $h$.

If $h=1$, then $\Gamma \Imp \Delta$ is an initial sequent in $G_\msf{PAL}$. By definition, it is also an initial sequent in $G_\msf{EL}$.

If $h>1$, we consider the last rule $\mc{R}$ applied in the derivation. If $\mc{R}$ is not a reduction rule, the claim can be proved by first applying the induction hypothesis to the premise(s) and then applying $\mc{R}$. 

If $\mc{R}$ is a reduction rule for (R1), (R2), (R4) or (R5), we apply the induction hypothesis to the premise(s), and then apply $(\to \Imp)$ or $(\Imp \to)$. We illustrate this by a few cases.

If $\mc{R}$ is $(R 1\Imp)$, let the derivation end with:
\[
\AxiomC{$ \Gamma' \Imp \Delta , x:\varphi$}
\AxiomC{$x:p,\Gamma' \Imp \Delta$}
\ri{\f$(R1\Imp)$}
\BinaryInfC{$x:\p {\varphi} p,\Gamma' \Imp \Delta$}
\DisplayProof
\]
By the induction hypothesis, we have $G_\msf{EL}\vdash_{h-1} t(\Gamma') \Imp t(\Delta) , x:t(\varphi)$ and $G_\msf{EL}\vdash_{h-1}x:t(p),t(\Gamma') \Imp t(\Delta)$. We proceed as follows:
\[
\AxiomC{$t(\Gamma') \Imp t(\Delta) , x:t(\varphi)$}
\AxiomC{$x:t(p),t(\Gamma') \Imp t(\Delta)$}
\ri{\f$(\to\Imp)$}
\BinaryInfC{$x:t(\varphi) \to t(p),t(\Gamma') \Imp t(\Delta)$}
\DisplayProof
\]
Since $t(x:\p{\varphi}p)= x:t(\varphi) \to t(p)$, we have 
\[
G_\msf{EL}\vdash_h  t(x:\p{\varphi}p), t (\Gamma')\Imp t(\Delta) .
\]

If $\mc{R}$ is $(\Imp R2)$, let the derivation end with:
\[
\AxiomC{$x:\varphi, \Gamma \Imp \Delta' , x:\neg \p{\varphi}\psi $}
\ri{\f$(\Imp R2)$}
\UnaryInfC{$\Gamma \Imp \Delta',x :\p{\varphi}\neg \psi $}
\DisplayProof
\]
By the induction hypothesis, we have $G_\msf{EL}\vdash_{h-1}x:t(\varphi), t(\Gamma)\Imp t(\Delta') , x:t(\neg \p{\varphi}\psi )$. We proceed as follows:
\[
\AxiomC{$x:t(\varphi), t(\Gamma)\Imp t(\Delta') , x:t(\neg \p{\varphi}\psi )$}
\ri{\f$(\Imp \to)$}
\UnaryInfC{$t(\Gamma) \Imp t(\Delta'),x :t(\varphi) \to t(\neg \p{\varphi}\psi )$}
\DisplayProof
\]
Since $t(x:\p{\varphi}\neg \psi ) = x:t(\varphi) \to t(\neg \p{\varphi}\psi )$, we have 
\[
G_\msf{EL}\vdash_h t (\Gamma)\Imp t(\Delta'), t(x:\p{\varphi}\neg p) .
\]

If $\mc{R}$ is $(R 5\Imp)$, let the derivation end with:
\[
\AxiomC{$ \Gamma' \Imp \Delta , x:\varphi$}
\AxiomC{$x: K_a\p{\varphi} \psi,\Gamma' \Imp \Delta$}
\ri{\f$(R5\Imp)$}
\BinaryInfC{$x:\p {\varphi} K_a \psi,\Gamma' \Imp \Delta$}
\DisplayProof
\]
By the induction hypothesis, we have $G_\msf{EL}\vdash_{h-1} t(\Gamma') \Imp t(\Delta) , x:t(\varphi)$ and $G_\msf{EL}\vdash_{h-1}x: t(K_a\p{\varphi}  \psi),t(\Gamma') \Imp t(\Delta)$. We proceed as follows:
\[
\AxiomC{$t(\Gamma') \Imp t(\Delta) , x:t(\varphi)$}
\AxiomC{$x:t(K_a\p {\varphi}  \psi),t(\Gamma') \Imp t(\Delta)$}
\ri{\f$(\to\Imp)$}
\BinaryInfC{$x:t(\varphi) \to t( K_a\p{\varphi}  \psi),t(\Gamma' )\Imp t(\Delta)$}
\DisplayProof
\]
Since $t(x:\p{\varphi}K_a \psi) = x:t(\varphi) \to t(K_a  \p{\varphi}\psi)$, we have 
\[
G_\msf{EL}\vdash_h t(x:\p{\varphi}K_a \psi), t (\Gamma')\Imp t(\Delta) .
\]

If $\mc{R}$ is a reduction rule for (R3), we apply the induction hypothesis to the premise(s), and then apply an inference rule for $\land$. We illustrate this by the case where $\mc{R}$ is $(R3\Imp)$.

If $\mc{R}$ is $(R3\Imp)$, let the derivation end with:
\[
\AxiomC{$x:\p{\varphi}\psi_1,x:\p{\varphi}\psi_2,\Gamma' \Imp\Delta $}
\ri{\f$(R3 \Imp)$}
\UnaryInfC{$x: \p{\varphi}(\psi_1\land \psi_2),\Gamma' \Imp \Delta$}
\DisplayProof
\]
By induction hypothesis, we have $G_\msf{EL}\vdash_{h-1} x:t(\p{\varphi}\psi_1), x:t( \p{\varphi}\psi_2), t(\Gamma') \Imp t(\Delta) $. We proceed as follows:
\[
\AxiomC{$x:t(\p{\varphi}\psi_1), x:t( \p{\varphi}\psi_2), t(\Gamma') \Imp t(\Delta) $}
\ri{\f$(\land \Imp)$}
\UnaryInfC{$ x:t(\p{\varphi}\psi_1)\land t( \p{\varphi}\psi_2), t(\Gamma') \Imp t(\Delta)$}
\DisplayProof
\]
Since $t(x:\p{\varphi}(\psi_1\land \psi_2)) =x: t(\p{\varphi}\psi_1)\land t( \p{\varphi}\psi_2)$, we have 
\[
G_\msf{EL}\vdash_{h} t({x:{\p{\varphi}(\psi_1\land \psi_2)}}), t(\Gamma') \Imp t(\Delta) .
\]

If $\mc{R}$ is a reduction rule for (R6), then we simply apply the induction hypothesis to the premise.

Thus completes the proof.
\end{proof}

\begin{rem}
Since $\msf{PAL}$ is more succinct than $\msf{EL}$ (see \cite{lutz2006complexity,french2013succinctness}), one might expect that a derivation of a sequent $\Gamma\Imp\Delta$ in $G_\msf{PAL}$ is strictly shorter (with respect to derivation height) than a derivation of $t(\Gamma)\Imp t(\Delta)$ in $G_\msf{EL}$. This is not the case for $G_\msf{PAL}$ and $G_\msf{EL}$ as shown in Theorem \ref{thm:10}. This is because a derivation for $\Gamma\Imp\Delta$ in $G_\msf{PAL}$ is obtained by executing the translation function $t$ (encoded by the reduction rules) in a derivation for $t(\Gamma)\Imp t(\Delta)$ in $G_{\msf{EL}}$, which increases the length of the derivation.
\end{rem}

\begin{cor}\label{thm:rulesPAL}
The following structural rules are admissible in $G_\msf{PAL}$:
\begin{align*}
&(w\Imp)~\frac{\Gamma \Imp \Delta}{x:\varphi, \Gamma \Imp \Delta}
&
\quad
&(\Imp w)~\frac{\Gamma \Imp \Delta}{\Gamma \Imp \Delta, x:\varphi}
\\
&(c\Imp)~\frac{x:\varphi,x:\varphi,\Gamma \Imp \Delta}{x:\varphi,\Gamma \Imp \Delta}
&\quad
&(\Imp c)~\frac{\Gamma \Imp \Delta,x:\varphi,x:\varphi}{\Gamma \Imp \Delta, x:\varphi}
\\
&(c_R\Imp)~\frac{x\sim_ay,x\sim_ay,\Gamma \Imp \Delta}{x\sim_ay,\Gamma \Imp \Delta}
&\quad
&(\Imp c_R)~\frac{\Gamma \Imp \Delta,x\sim_ay,x\sim_ay}{\Gamma \Imp \Delta, x\sim_ay}
\\
&(Cut)~\frac{\Gamma \Imp \Delta, x:\varphi \quad x:\varphi ,\Gamma' \Imp \Delta'}{\Gamma ,\Gamma'\Imp \Delta, \Delta'}
\end{align*}
\end{cor}

\begin{proof}
It follows directly from Theorem \ref{thm:adm} and Theorem \ref{thm:10}. Take $(Cut)$ as an example. Assume that $G_\msf{PAL}\vdash \Gamma \Imp \Delta, x:\varphi$ and $x:\varphi ,\Gamma' \Imp \Delta'$. By Theorem \ref{thm:10}, $G_\msf{EL}\vdash t(\Gamma) \Imp t(\Delta) , x: t(\varphi)$ and $G_\msf{EL} \vdash x:t(\varphi) ,t(\Gamma') \Imp t(\Delta')$. Then by admissibility of cut in $G_\msf{EL}$ (Theorem \ref{thm:adm}), $G_\msf{EL} \vdash t(\Gamma),  t(\Gamma') \Imp t(\Delta), t(\Delta')$. By Theorem \ref{thm:10} again, $G_\msf{PAL} \vdash \Gamma,\Gamma' \Imp \Delta, \Delta'$.
\end{proof}

\begin{thm}[Soundness and Completeness]
For any $\mc{L}_\msf{PAL}$-formula $\varphi$, $\varphi \in \msf{PAL}$ iff $G_\msf{PAL}\vdash \Imp \varphi$.
\end{thm}

\begin{proof}
For the left-to-right direction, it suffices to show that each axiom in the Hilbert-style axiomatization for $\msf{PAL}$ is derivable in $G_\msf{PAL}$ and that (MP) and (GK$_a$) are admissible in $G_\msf{PAL}$. Since $G_\msf{PAL}$ is an extension of $G_\msf{EL}$, derivability of axioms in $\msf{EL}$. Admissibility of (MP) follows from admissibility of cut in $G_\msf{PAL}$ (Corollary \ref{thm:rulesPAL}). Admissibility of (GK$_a)$ in $G_\msf{PAL}$ follows from admissibility of (GK$_a)$ in $G_\msf{EL}$ and Theorem \ref{thm:10}. For the reduction axioms, since each reduction axiom has a corresponding pair of sequent rules in $G_\msf{PAL}$, their derivations can be obtained by direct root-first search. One derivation of (R5) is given in Example \ref{exa:0414}.

For the right-to-left direction, it suffices to show that each sequent rule in $G_\msf{PAL}$ is valid. This is routine. We omit the proof here.
\end{proof}

\section{Decidability}

In this section we show that $G_\msf{PAL}$ allows terminating proof search. This result can be proved indirectly by Theorems \ref{thm:ElDeci} and \ref{thm:10}. In this section we present a direct proof.

Readers familiar with the proof for terminating proof search in $G_\msf{EL}$ may notice that the proof for terminating proof search in $G_\msf{PAL}$ is essentially the same as that in $G_\msf{EL}$ because $G_\msf{PAL}$ is an extension of $G_\msf{EL}$ with some reduction rules. For this reason, we sketch the proof in this section and refer to \cite{negri2005proof} for a detailed description.

To show that $G_\msf{PAL}$ allows terminating proof search, first we extend the notion of \emph{subformulas} for $\mc{L}_\msf{PAL}$-formulas, because active formulas are not subformulas of the principal formula in the usual sense in the reduction rules of $G_\msf{PAL}$.

\begin{defn}
Let $\varphi$ be a $\mc{L}_\msf{PAL}$-formula. Formula $\psi$ is called a \emph{semi-subformula} of $\varphi$ if one of the following conditions hold:
\begin{enumerate}
\item $\psi$ is a subformula of $\varphi$;
\item $\varphi= \p{\phi}\ast \chi$ and $\psi =\ast \p{\phi}\chi$, where $\ast \in \{\neg, K_a\}$;
\item $\varphi= \p{\phi}( \chi_1\star \chi_2)$ and $\psi =\p{\phi} \chi_1$ or $\psi =\p{\phi}\chi_2$, where $\star \in \{\land, \to\}$; 
\item $\varphi = \p{\phi}\p{\chi}\xi$ and $\psi = \p{\phi \land \p{\phi}\chi} \xi$.
\end{enumerate}
The set of semi-subformulas of $\varphi$ is denoted by $SSub(\varphi)$. We say that $\psi$ is a \emph{proper semi-subformula} of $\varphi$ if $\psi$ is a semi-subformula of $\varphi$ and $\psi \not =\varphi$.
\end{defn}
 By Definition \ref{defn:compl}, the complexity of a proper semi-subformula of $\varphi$ is strictly smaller than that of $\varphi$.

\begin{defn}
A derivation in $G_\msf{PAL}$ satisfies the \emph{weak subformula property} if all expressions in the derivation are either semi-subformulas of formulas $\varphi$ in labelled formulas $x:\varphi$ in the endsequent of the derivation, or relational atoms of the form $x\sim_ay$.
\end{defn}

It is easy to verify that each derivation in $G_\msf{PAL}$ satisfies the weak subformula property.

The potential sources of non-terminating proof search in $G_\msf{PAL}$ are as follows:
\begin{enumerate}[(1)]
\item $(Trans_a)$, $(Sym_a)$ (read root-first) can be applied infinitely on the same principal formulas.
\item $(Ref_a)$ (read root-first) can be applied infinitely to introduce new relational atoms. 
\item ${(K_a\Imp)}$ (read root-first) can be applied infinitely on the same principal formulas.
\item By $(Trans_a)$ and its iteration with $( \Imp K_a)$ that brings in new accessible worlds, we can build chains of accessible worlds on which $(K_a\Imp)$ can be applied infinitely.
\end{enumerate}
 To show that the space of root-first proof search is finite, it is useful to look at \emph{minimal derivations}, that is, derivations where shortenings are not possible. 
 
For (1), since no rules, read root first, can change a relational atom, any application of $(Trans_a)$ or $(Sym_a)$ on the same principal formulas will make the derivation fail to be minimal. Therefore, this case should be excluded when searching for minimal derivation.

For (2),  we need the following lemma:

\begin{lem}
All variables in relational atoms of the form $x\sim_ax$ removed by $(Ref_a)$ in a minimal derivation of  a sequent $\Gamma \Imp \Delta$ in $G_\msf{PAL}$ are variables in $\Gamma$ or $\Delta$.
\end{lem}

\begin{proof}
It can be proved by tracing the relational atom $x\sim_ax$ up in the derivation. For a detailed proof, see Lemma 6.2 of \cite{negri2005proof}. 
\end{proof}

For (3), we need to show the following lemma:

\begin{lem}
Rule $(K_a\!\Imp)$ permutes down with respect to $(\Imp\!K_a)$ in case the principal relational atom of $(K_a\!\Imp)$ is not active in $(\Imp\!K_a)$. Moreover, rule $(K_a\!\Imp)$ permutes down with respect to all other rules.
\end{lem}

Then we can prove the following proposition:

\begin{prop}
$(K_a \! \Imp)$ and $(\Imp \! K_a)$ cannot be applied more than once on the same pair of principal formulas on any branch in any derivation in $G_\msf{PAL}$.
\end{prop}

For (4) we need to prove the following proposition:
\begin{prop}
In a minimal derivation of a sequent in $G_\msf{PAL}$, for each formula $x:K_a \varphi$ in its positive part there are at most $n(K_a)$ applications of $(\Imp K_a)$ iterated on a chain of accessible worlds $x\sim_ax_1, x_1\sim_ax_2,\ldots$ with principal formula $x_i: K_a \varphi$, where $n(K_a)$ denotes the number of $K_a$ in the negative part of $\Gamma \Imp \Delta$.
\end{prop}

\section{Conclusion, comparison and future research}

In this paper we proposed a labelled sequent calculus for $\msf{PAL}$. It is based on a proof system for $\msf{EL}$, namely poly-modal $\msf{S5}$, extended with sequent rules to deal with announcement operators that directly mirror the $\msf{PAL}$ reduction rules. We also determined that the obtained system allows terminating proof search, and compared our system with the calculus for $\msf{EL}$ on matters such a height preservation of derivations. 

Various proposals for sequent calculi for PAL have been made in \cite{balbiani2014sequent,BAL2010,maffezioli2010gentzen,NOM2016}. The labelled sequent calculi in \cite{balbiani2014sequent,NOM2016} for $\msf{PAL}$ are based on a reformulation of the semantics for $\msf{PAL}$. 
The notion of a model restriction to a (single) formula is generalized in these works to that of a model restriction to a list of formulas. Denote by $\alpha$ or $\beta$ finite lists $(\varphi_1, \ldots, \varphi_n)$ of formulas, and by $\epsilon$ the empty list. For any list $\alpha =(\varphi_1,\ldots, \varphi_n)$ of formulas, define $\mc{M}^\alpha$ recursively as follows: $\mc{M}^\alpha: = \mc{M}$ (if $\alpha = \epsilon)$, and $\mc{M}^\alpha := (\mc{M}^\beta)^{\varphi_n} = (W^{\beta, \varphi_n}, (\sim^{\beta, \varphi_n}_a)_{a \in \msf{Ag}},V^{\beta, \varphi_n})$ (if $\alpha =\beta, \varphi_n$).

Then an equivalent semantics for $\msf{PAL}$ is defined as follows:
\[
\begin{array}{lll}
\mathcal{M}^{\alpha,\varphi},w \Vdash p &~\text{iff}~ & \mathcal{M}^\alpha, w\Vdash \varphi \text{~and~} \mc{M}^\alpha, w\Vdash p
\\
\mathcal{M}^\alpha,w \Vdash \neg \varphi      &~\text{iff}~& \mathcal{M}^\alpha,w \nVdash \varphi
\\
\mathcal{M}^\alpha,w \Vdash \varphi \land \psi   &  ~\text{iff}~          & \mathcal{M}^\alpha,w \Vdash \varphi        ~\text{and}~         \mathcal{M}^\alpha,w \Vdash \psi
\\
\mathcal{M}^\alpha,w \Vdash \varphi \to \psi   &  ~\text{iff}~          & \mathcal{M}^\alpha,w \nVdash \varphi        ~\text{or}~         \mathcal{M}^\alpha,w \Vdash \psi
\\
\mathcal{M}^\alpha,w \Vdash K_{a}\varphi & ~\text{iff}~ & \text{for~all}~ v  \in W, w\sim^\alpha_{a}v ~\text{implies}~ \mathcal{M}^\alpha,v \Vdash \varphi
\\
\mathcal{M}^\alpha,w \Vdash \p{\varphi}\psi  &~\text{iff}~  & \mathcal{M}^\alpha,w \Vdash \varphi  ~\text{implies}~ \mathcal{M}^{\alpha,\varphi},w \Vdash \psi
\end{array}
\]

With this semantics, Balbiani \cite{balbiani2014sequent} and Nomura et al. \cite{NOM2016}\footnote{Nomura et al. \cite{nomura2020cut} extended the method developed in \cite{NOM2016} to action modal logic.}
 developed different labelled sequent calculi for $\msf{PAL}$, simultaneously repairing some perceived deficiencies in the previously published in Maffezioli and Negri \cite{maffezioli2010gentzen}. The calculus in Balbiani \cite{balbiani2014sequent} admits cut and allows terminating proof search, and the calculus in Nomura et al. \cite{NOM2016} also admits cut. 
Note that these calculi, unlike ours, are for $\msf{PAL}$ based on the normal modal logic $\msf{K}$, whereas we took $\msf{PAL}$ as an extension of  $\msf{S}5$ (like \cite{hakli2007proof,maffezioli2010gentzen}). We therefore also include the usual inference rules for $\msf{S}5$ (i.e., $(Ref_a)$, $(Trans_a)$ and $(Sym_a)$) into our calculus. Further, our calculus uses sequent rules based on $\msf{PAL}$ reduction axioms to deal with the announcement operators.


This method can be directly applied to action model logic \cite{DIT2007} since this logic is also an extension of $\msf{EL}$ with some reduction axioms, of which PAL is a special case.

There are many extensions of other logics with public announcement that also involve  reduction axioms for such announcements, such as intuitionistic PAL \cite{ma2014algebraic}, bilattice PAL \cite{rivieccio2014algebraic}, {\L}ukasiewicz PAL \cite{LucasiPAL}, and variations/extensions of action model logic like bilattice logic of epistemic actions and knowledge \cite{bakhtiari2020bilattice}. We can try to develop calculi for these logics with the steps similar to those for $\msf{PAL}$.


\bibliographystyle{splncs04}
\bibliography{ref}

\end{document}